%% file: Ducoffe.tex
\documentclass[11pt]{article}

\usepackage{amsmath,amsthm,amssymb}
\usepackage{hyperref}
\usepackage{authblk}
\usepackage{fullpage}

\newtheorem{theorem}{Theorem}
\newtheorem{lemma}{Lemma}
\newtheorem{corollary}{Corollary}

\theoremstyle{definition}

\begin{document}

\title{Eccentricity queries and beyond using Hub Labels}

\author[1,2]{Guillaume Ducoffe}
\affil[1]{\small National Institute for Research and Development in Informatics, Romania}
\affil[2]{\small University of Bucharest, Romania}
\date{}

\maketitle

\begin{abstract}
    \input{abstract}
\end{abstract}

\section{Introduction}\label{sec:intro}
\input{introduction}

\section{Properties of Hub Labels}\label{sec:hub}
\input{hub-label}

\section{Applications}\label{sec:applications}
\input{applications}

\section{Graph classes of bounded expansion}\label{sec:nowhere-dense}
\input{expansion}

\bibliographystyle{abbrv}
\bibliography{hub}

\end{document}

%% file: abstract.tex
{\em Hub labeling schemes} are popular methods for computing distances on road networks and other large complex networks, often answering to a query within a few microseconds for graphs with millions of edges. 
In this work, we study their algorithmic applications beyond distance queries. Indeed, several implementations of hub labels were reported to have good practical performances, both in terms of pre-processing time and maximum label size. There are also a few relevant graph classes for applications for which we know how to compute hub labelings with sublinear labels in quasi linear time. These positive results raise the question of what can be computed efficiently given a graph and a hub labeling of small maximum label size as input. We focus on {\em eccentricity queries} and {\em distance-sum queries}, for several versions of these problems on directed weighted graphs, that is in part motivated by their importance in facility location problems. On the negative side, we show conditional lower bounds for these above problems on unweighted undirected sparse graphs, via standard constructions from ``Fine-grained'' complexity. Specifically, under the Strong Exponential-Time Hypothesis (SETH), answering to these above queries requires $\Omega(|V|^{2-o(1)})$ pre-processing time or  $\Omega(|V|^{1-o(1)})$ query time, even if we are given as input a hub labeling of maximum label size in $\omega(\log{|V|})$. However, things take a different turn when the hub labels have a {\em sublogarithmic} size. Indeed, given a hub labeling of maximum label size $\leq k$, after pre-processing the labels in total $2^{{O}(k)} \cdot |V|^{1+o(1)}$ time, we can compute both the eccentricity and the distance-sum of any vertex in $2^{{O}(k)} \cdot |V|^{o(1)}$ time. Our data structure is a novel application of the orthogonal range query framework of (Cabello and Knauer, {\it Computational Geometry}, 2009). It can also be applied to the fast {\em global} computation of some topological indices. Finally, as a by-product of our approach, on any fixed class of unweighted graphs with {\em bounded expansion}, we can decide whether the diameter of an $n$-vertex graph in the class is at most $k$ in $f(k) \cdot n^{1+o(1)}$ time, for some ``explicit'' function $f$. This result is further motivated by the empirical evidence that many classes of complex networks have bounded expansion (Demaine et al., {\it Journal of Computer and System Sciences}, 2019), and that some of these networks, such as the online social networks, have a relatively small diameter.

%% file: introduction.tex
We refer to~\cite{BoM08,Die12} for basic notions and terminology in Graph Theory.
Real-world networks can be conveniently represented as a weighted graph $G=(V,E,w)$, with typical interpretations of the weight function $w : E \to \mathbb{R}_{> 0}$ being the latency of a link, the length of a road between two locations, etc. If $\forall e \in E, w_e = 1$, then we can drop the weight function, calling the graph unweighted and writing $G=(V,E)$. Unless stated otherwise, all graphs considered in this note are finite, weighted, directed and strongly connected\footnote{The strongly connected assumption could be easily dropped, but it makes the presentation simpler by avoiding the pathological case of infinite distances.}. For every ordered pair of vertices $(u,v)$, the distance $dist_G(u,v)$ is defined as the smallest weight of a path from $u$ to $v$. We omit the subscript if the graph $G$ is clear from the context. We note that distance computation is an important primitive in many scenarios, ranging from Satellite Navigation devices to forwarding network packets. However, in some practical or theoretical applications, the ``right'' notion of distance is a twist of the one presented above. Let us give a few examples (taken from~\cite{AVW16}).
\begin{itemize}
\item{\bf Source distance:} $dist(u,v)$. This is the standard distance.

\item{\bf Max distance:} $\max\{dist(u,v),dist(v,u)\}$. It turns the distances in a directed graph into a metric ({\it i.e.}, it is symmetric and it satisfies the triangular inequality).
 
\item{\bf Min distance:} $\min\{dist(u,v),dist(v,u)\}$. As a concrete application of the latter, Abboud et al.~\cite{AVW16} cited the problem of optimally locating a hospital (minimizing the time for either the patient coming to the hospital, or an ambulance drive being sent at the patient's home).
 
\item{\bf Roundtrip distance:} $dist(u,v)+dist(v,u)$. It also turns the distances in a directed graph into a metric. Many distance computation techniques for {\em undirected} graphs can be generalized to roundrip distances, such as {\it e.g.}, the existence of low-stretch sparse spanners~\cite{RTZ08}.
\end{itemize}

A {\em distance oracle} is a data structure for computing distances. Next, we present a well-known family of distance oracles, of which we explore in this work the algorithmic properties.

\medskip
\noindent
{\it Hub labeling.} A $2$-hopset or, as we call it here, a hub labeling, assigns to every vertex $v$ an ordered pair $(L^+(v),L^-(v))$ of vertex-subsets along with the distances $dist(v,x)$, for $x \in L^+(v)$, and $dist(x,v)$, for $x \in L^-(v)$. This labeling must ensure that, for any ordered pair of vertices $(u,v)$, we have $dist(u,v) = \min \{ dist(u,x) + dist(x,v) \mid x \in L^+(u) \cap L^-(v) \}$. In particular, we can compute $dist(u,v)$ -- and more generally, the source, max-, min- and roundtrip distances from $u$ to $v$ -- in ${O}(|L^+(u)|+|L^-(v)|)$ time. It has been observed in~\cite{KUL19} that most distance labeling constructions are based on hub label schemes, thereby further motivating our study of the latter as a first step toward a more general analysis of the algorithmic properties of distance oracles.

\smallskip
\noindent
The maximum label size is defined as $\max \{ |L^+(v)| \mid v \in V \} \cup \{ |L^-(v)| \mid v \in V \}$, and it is the main complexity measure for hub labelings\footnote{The {\em bit size} is larger by at least a logarithmic factor, because we need to encode the vertices' IDs and the distances values for all the vertices inside the labels.}. Every graph admits a hub labeling of maximum label size $\Theta(n)$, and this is essentially optimal~\cite{GPPR0404}. However, it has been experimentally verified that many real-world graphs, and especially road networks, admit hub labelings with small maximum label size~\cite{CHKZ03}. From a more theoretical point of view, every $n$-node planar graph admits a hub labeling with maximum label size ${O}(\sqrt{n}\log{n})$~\cite{GPPR0404}. Bounded-treewidth graphs and graphs of bounded {\em highway dimension} -- a parameter conjectured to be small for road networks -- admit hub labelings with (poly)logarithmic maximum label size~\cite{ADFG+16,GPPR0404}. Graphs of bounded tree-depth even have hub labelings with {\em constant} maximum label size, that can be computed in quasi linear time~\cite{IOO18}. 

\medskip
\noindent
{\it Eccentricity and Distance-sum computations.} Since many hub labeling schemes are readily implemented~\cite{AIY13,ADGW11,ADGW12,DGSW14,LUYK17}, our approach consists in assuming a graph to be given with a hub labeling, and to study whether this additional information helps in solving faster certain distance problems. Specifically, let $D(\cdot,\cdot)$ represent any distance function on directed graphs (source, min-, max- or roundtrip).  The {\em eccentricity} of a vertex $u$, denoted by $e^D(u)$, is the maximum distance from $u$ to any other vertex, {\it i.e.}, $e^D(u) = \max_v D(u,v)$. The {\em diameter} and the radius of a graph -- w.r.t. $D$ -- are, respectively, the maximum and minimum eccentricities amongst its vertices. The {\em distance-sum} of a vertex $u$, that we denote by $s^D(u)$, is the sum of all the distances from $u$ to any other vertex, {\it i.e.}, $s^D(u) = \sum_v D(u,v)$. A median is a vertex minimizing its distance-sum, while the Wiener index is equal to $\sum_u s^D(u)$. We stress that both computing the radius and the median set of a graph are fundamental {\em facility location} problems: with applications to optimal placement of schools, hospitals and other important facilities on a network. The problems of computing the diameter and the Wiener index found applications in network design and chemistry~\cite{RoK02}, respectively. In Sec.~\ref{sec:applications}, we consider other topological indices from chemical graph theory.

\subsection{Related work}

Computing the eccentricities in a graph is a notoriously ``hard'' problem in P, for which no quasi linear-time algorithm is likely to exist. For {\em weighted} $n$-vertex graphs, the problem is subcubic equivalent to All-Pairs Shortest-Paths (APSP). It is conjectured that we cannot solve APSP in ${O}(n^{3-\varepsilon})$ time for any $\varepsilon > 0$~\cite{AGW14}. For the special case of {\em unweighted} graphs, we can compute the eccentricities in truly subcubic time using fast matrix multiplication~\cite{Sei95,Yus10}. However if we assume the Strong Exponential-Time Hypothesis (SETH), then it cannot be improved to ${O}(n^{2-\varepsilon})$ time, for any $\varepsilon > 0$~\cite{RoV13}. Under a related hypothesis, the so-called ``Hitting Set Conjecture'', similar hardness results were proved for median computation~\cite{AVW16}. For all that, some heuristics and exact algorithms for diameter, radius and eccentricities computations do perform well in practice~\cite{BCHK+15,CGHL+13,Shu15,TaK13}, which has been partially justified from a theoretical viewpoint~\cite{DHV18+,Wag19,Wag20}.

\smallskip
\noindent
These hardness results have renewed interest in characterizing the graph classes for which we can compute some of these above distance problems in truly subquadratic time, or even  in quasi linear time. In this respect, several recent results were proved by using classical tools from Computational Geometry~\cite{AVW16,BHM20,CaK09,Cab18,DuD19+,DHV20,DHV19,GKHM+18}. Perhaps the most impressive such result is the truly subquadratic algorithm of Cabello for computing the diameter and the Wiener index of planar graphs~\cite{Cab18}. It is worth of mentioning that Cabello's techniques were the starting point of recent breakthrough results for exact {\em distance oracles} on planar graphs~\cite{CDW17,GMWW18}. In this paper, we ask whether conversely, having at hands an exact distance oracle with certain properties can be helpful in the design of fast algorithms for some distance problems on graphs. We give a partial answer for hub labelings.

\paragraph{Hub Labelings.}
Computing a hub labeling with optimal maximum label size is NP-hard~\cite{BGKS+15}, but it can be approximated up to a logarithmic factor in polynomial time~\cite{CHKZ03}.
The use of hub labels beyond distance queries was already considered in~\cite{EfEP15} and~\cite{ZhY20} for $k$-nearest neighbours queries and shortest-path counting, respectively.
%
Several generalizations of hub labels were proposed.
For instance, a {\em $3$-hopset} assigns a pair of labels $(L^-(v),L^+(v))$ for every vertex $v$ (as for hub labelings), but it computes an additional {\em global} set of distances, denoted by $E_2$. This labeling must ensure that, for every vertices $u$ and $v$, $dist(u,v) = \min\{ dist(u,x) + dist(x,y) + dist(y,v) \mid x \in L^+(u), y \in L^-(v), (x,y) \in E_2 \}$~\cite{GKV19}.
A {\em distance labeling scheme} is given by an encoding function $L$ and a decoding function $D$.
It assigns a label $L(v)$ to every vertex $v$ s.t., for every vertices $u$ and $v$, we have $dist(u,v) = D(L(u),L(v))$~\cite{GPPR0404}.
These extensions are to be investigated in future work.

\subsection{Results}

We first address the following types of queries on the vertices $v$ of a graph $G$.
\begin{itemize}
\item{\bf Eccentricity query:} compute the eccentricity of $v$; 
\item{\bf Distance-sum query:} compute the distance-sum of $v$.
\end{itemize}
Before stating our results, let us discuss about naive resolutions for these two types of queries.
If we compute APSP in ${O}(|V|^3)$ time, then we can compute all eccentricities and distance-sums in additional ${O}(|V|^2)$ time. In doing so, we get a trivial data structure in ${O}(|V|)$-space and ${O}(1)$ query time. Therefore, the main challenge here is to decrease the pre-processing time.
A hub labeling (or more generally, a distance-labeling) of maximum label size $k$ is an implicit representation of the distance-matrix in ${O}(k|V|)$ space.
It allows us to answer these above two types of queries in ${O}(k|V|)$ time.
However, this query time is in $\Omega(|V|)$, even in the very favourable case of $k = {O}(1)$.
If for instance, the graph considered is sparse, then this is no better than computing a shortest-path tree from scratch.   
%
Our first observation is that, {\it in general}, hub labelings do {\it not} help. Namely:  

\begin{theorem}\label{thm:seth}
Under SETH, for any $\varepsilon > 0$, any data structure for answering eccentricity queries or distance-sum queries requires $\Omega(|V|^{2-\varepsilon})$ pre-processing time or $\Omega(|V|^{1-\varepsilon})$ query time.
This holds even if a hub labeling of maximum label size $\omega(\log{|V|})$ is given as part of the input. 
\end{theorem}

\begin{proof}
We consider undirected graphs, for which all the aforementioned distances are equivalent (up to a factor $2$ for the roundtrip distance). In particular, a split graph is a graph whose vertex-set can be bipartitioned into a clique $K$ and a stable set $S$. We observe that any unweighted split graph $G=(K\cup S,E)$ admits a hub labeling with maximum label size ${O}(|K|)$, that can be computed in ${O}(|K||S|)$ time: simply store, for each vertex, its distance from and to any vertex of $K$, that is either $1$ or $2$. If we further assume $|K|=n^{o(1)}$, then this trivial hub labeling can be computed in $n^{1+o(1)}$ time.

\smallskip
\noindent
Under SETH, for any $\varepsilon > 0$, we cannot compute the diameter of $n$-vertex split graphs in ${O}(n^{2-\varepsilon})$ time, and this result holds even if $|K| = \omega(\log{n})$~\cite{BCH16}. 
Assume the existence of a data structure for eccentricity queries, with $P(n)$ pre-processing time and $Q(n)$ query time. In particular, we can compute the diameter in $P(n) + n \cdot Q(n)$ time. Therefore, $P(n) = \Omega(n^{2-\varepsilon})$, or $Q(n) = \Omega(n^{1-\varepsilon}).$

\smallskip
\noindent
In the same way, assume the existence of a data structure for distance-sum queries, with $P'(n)$ pre-processing time and $Q'(n)$ query time. Observe that the diameter of a non-complete split graph is either $2$ or $3$. It is folklore (see, {\it e.g.},~\cite{BHM20}), that an undirected unweighted graph has diameter at most $2$ if and only if the distance-sum of any vertex of degree $d$ is equal to $2(n-1) - d$. In particular, we can compute the diameter of an unweighted undirected split graph in $P'(n) + n \cdot Q'(n)$ time.  As a result, under SETH we must also have $P'(n) = \Omega(n^{2-\varepsilon})$, or $Q'(n) = \Omega(n^{1-\varepsilon}).$ 
\end{proof}

\noindent
{\it A framework for fast queries computation.} Positive results can be derived for hub labelings with sublogarithmic labels. Our algorithm is a novel application of a popular framework for fast diameter computation, using orthogonal range queries~\cite{AVW16,BHM20,CaK09,Duc19,DHV19}.

\begin{theorem}\label{thm:hub-label}
For source, min-, max- and roundtrip distances, for every graph $G=(V,E,w)$, if we are given a hub labeling with maximum label size $\leq k$, then we can compute a data structure for answering eccentricity queries and distance-sum queries with $2^{{O}(k)} \cdot |V|^{o(1)}$ query time. This takes $2^{{O}(k)} \cdot |V|^{1+o(1)}$ pre-processing time.
\end{theorem}

We stress that in practice, the bottleneck of Theorem~\ref{thm:hub-label} is the computation of a hub labeling. For this task, we are bound to use heuristics~\cite{DGSW14}. In a few restricted classes, such as graphs of bounded tree-depth, {\em constant-size} hub labels can be computed in quasi linear time~\cite{IOO18}. -- Below, we give an application of this result to graph classes of bounded expansion. -- Let us cite, as another example of graphs with small hub labels, the graphs of bounded vertex-cover. Indeed, we can easily derive a hub labeling of maximum label size $k$ from a vertex-cover of cardinality $k$. However, for the {\em unweighted} graphs of bounded vertex-cover, there exist slightly faster methods for eccentricities and distance-sum computations than our Theorem~\ref{thm:hub-label}~\cite{BHM20,CDP19}.
We hope that our results will encourage the quest for other graph classes that admit constant-size hub labels. For instance, it was conjectured in~\cite{ADFG+16} that all graphs of highway dimension at most $h$ admit a hub labeling with maximum label size ${O}(h)$. Our results in this paper show that solving this conjecture in the affirmative would be a first step toward fast diameter computation within these graphs, with applications to road networks.

\smallskip
\noindent
Our techniques are inspired from those in~\cite{AGW14} for bounded treewidth graphs, of which they are to some extent a generalization. Indeed, the algorithms in~\cite{AGW14} for graphs of treewidth at most $k$ parse a hub labeling of maximum label size in ${O}(k\log{n})$, but where every label is composed of ${O}(\log{n})$ levels, each containing ${O}(k)$ vertices. Roughly, it allows the authors from~\cite{AGW14} to build an algorithm of logarithmic recursion depth, where at each call we may assume to be given hub labels of maximum size ${O}(k)$. However, in~\cite{AGW14}, all the nodes considered have the same ${O}(k)$ vertices in their labels (the latter forming a balanced separator). This is no more true for general hub labels of maximum size ${O}(k)$, a case which requires more complex range queries than in~\cite{AGW14}. Incidentally, for {\em roundtrip distances}, we improved on our way the $2^{{O}(k^2)}n^{1+o(1)}$-time algorithm from~\cite{AGW14} for directed graphs of treewidth at most $k$: 

\begin{corollary}\label{cor:treewidth}
For the roundtrip distance, for any $n$-vertex directed graph of treewidth at most $k$, we can compute all the eccentricities and distance-sums in $2^{{O}(k)}n^{1+o(1)}$ time.
\end{corollary}

\noindent
{\it Application to graph classes of bounded expansion.} Our initial motivation for this work was to study constant diameter computation within graph classes of {\em bounded expansion}. Without entering too much into technical details, the graph classes of bounded expansion are exactly those whose so-called shallow minors are all sparse (see Sec.~\ref{sec:nowhere-dense} for a formal definition). In particular, it generalizes bounded-degree graphs, bounded-treewidth graphs, and more generally properly minor-closed graph classes~\cite{NeO12}. In~\cite{DRRV19}, evidence was given that many classes of complex networks exhibit a bounded-expansion structure. Among those networks, social networks are well-known to obey the so-called ``small-world'' property~\cite{WaS98}, that implies a relatively small diameter. In this context, our theoretical results for constant-diameter computation within graph classes of bounded expansion might be a first step toward more practical algorithms for computing the diameter of social networks and other complex networks with similar properties.

\begin{theorem}\label{thm:nowhere-dense}
For every class of graphs $\mathcal{G}$ of bounded expansion, for every $n$-vertex unweighted graph $G \in \mathcal{G}$ and positive integer $k$, we can decide whether the diameter of $G$ is at most $k$ in ${O}(f(k) \cdot n^{1+o(1)})$ time, for some function $f$.
\end{theorem}

Diameter computation within unweighted graph classes of bounded expansion has already received some attention in the literature. In particular, for the special case of undirected graphs of maximum degree $3$, Dahlgaard and Evald proved that under SETH, we cannot compute the diameter in truly subquadratic time~\cite{EvD16}. But their hardness results hold for bounded-degree graphs of super-logarithmic diameter. In contrast to this negative result, testing whether a graph has diameter at most some {\em constant} $k$ can be written as a first-order formula of size ${O}(k)$. Therefore, in any class of unweighted graphs of bounded expansion, we can derive from a prior work of Dvo\u{r}\'{a}k et al. a quasi linear-time {\em parameterized} algorithm for constant diameter computation~\cite{DKT10}. Unfortunately, the hidden dependency in the parameter $k$ is rather huge due to the use of Courcelle's theorem~\cite{Cou90}. Recently, we proposed a different approach for constant diameter computation within nowhere dense graph classes -- a broad generalization of the graph classes of bounded expansion--, based on a VC-dimension argument~\cite{DHV20}. However, the running time of our algorithm for deciding whether the diameter of an $n$-vertex graph is at most $k$ was of order $\tilde{O}(n^{2-f(k)})$ for some super-exponential function $f$.

Our Theorem~\ref{thm:nowhere-dense} improves on these previous works by using {\em low tree-depth decompositions}: a covering of $n$-vertex graphs of bounded expansion by relatively few subgraphs of bounded {\em tree-depth}, so that for a fixed $k$, each $k$-vertex subgraph is contained in at least one subgraph of this covering~\cite{NeO12}. Combined with previous results on hub labelings within bounded tree-depth graph classes~\cite{IOO18}, it allows us to prove the existence, for every fixed $k$, of ${O}(f(k))$-size hub labels for the pairs of vertices at distance $\leq k$, for some ``explicit'' function $f$ -- about a tower of exponentials of height four. Then, Theorem~\ref{thm:nowhere-dense} follows from our Theorem~\ref{thm:hub-label}. 
We left open whether our approach could be generalized to {\em nowhere dense} graph classes.

\subsection{Organization of the paper}
Our main technical contribution (Theorem~\ref{thm:hub-label}) is proved in Sec.~\ref{sec:hub}. In Sec.~\ref{sec:tw} and~\ref{sec:ti}, we discuss applications of Theorem~\ref{thm:hub-label} to, respectively, bounded-treewidth graphs (Corollary~\ref{cor:treewidth}) and the computation of topological indices. In particular, in Sec.~\ref{sec:ti}, we exploit recent results of Cabello from~\cite{Cab19}, and we present a new type of distance information which can be computed from the orthogonal range query framework. Finally, we conclude this paper in Sec.~\ref{sec:nowhere-dense} with our results for graph classes of bounded expansion (Theorem~\ref{thm:nowhere-dense}).

%% file: hub-label.tex
\subsection{Range Queries}\label{sec:range-query}

We first review some important terminology, and prior results. Let $V$ be a set of $k$-dimensional points, where each point $\overrightarrow{p} = (p_1,p_2,\ldots,p_k) \in V$ is assigned some value $f(\overrightarrow{p})$.
A box is the cartesian product of $k$ intervals $I_i = (l_i,u_i)$, for $1 \leq i \leq k$, denoted by $\mathcal{R} = I_1 \times I_2 \times \ldots \times I_k$. Note that we allow each of the $k$ intervals to exclude either of its ends, and that we allow these ends to be infinite. Furthermore, $\overrightarrow{p} \in \mathcal{R} \cap V$ if and only if $\forall 1 \leq i \leq k, \ p_i \in I_i$.
A {\em range query} asks for some information about the points within a given box.
We use the following types of range queries:
\begin{itemize}
\item{\bf Max-Query:} Compute a point $\overrightarrow{p} \in \mathcal{R}$ maximizing $f(\overrightarrow{p})$;
\item{\bf Sum-Query:} Compute the sum $\sum_{\overrightarrow{p} \in \mathcal{R}} f(\overrightarrow{p})$;
\item{\bf Count-Query:} Count $| V \cap \mathcal{R} |$ (the latter can be obtained from the above sum-query by setting $\forall \overrightarrow{p} \in V, f(\overrightarrow{p}) = 1$).
\end{itemize}
We refer to~\cite{BHM20} for a thorough treatment of range queries and their applications to distance problems on graphs. The $k$-dimensional range tree is a classic data structure in order to answer range queries efficiently for static point sets. Evidence of its practicality for graph problems was given in~\cite{Mag18}.

\begin{lemma}[\cite{BHM20}]\label{lem:range-tree}
For all the aforementioned types of range queries, for every $k$-dimensional point set $V$ of size $n$, we can construct a $k$-dimensional range tree in $2^{{O}(k)} n^{1+o(1)}$ time, that allows to answer a query in $2^{{ O}(k)}n^{o(1)}$ time. 
\end{lemma}

\subsection{Proof of Theorem~\ref{thm:hub-label}}

We are now ready to prove our main algorithmic tool for this paper. Our approach is essentially independent from the type of query and the distance considered. We first present our results for eccentricity queries, postponing the slight changes to be made for distance-sum queries until the end of this section. Similarly, we postpone the specific parts for each distance function to the end of the section. In what follows, let $D(\cdot,\cdot)$ be an arbitrary distance function ({\it i.e.}, source, min-, max- or roundtrip).

\medskip
\noindent
{\it Decomposition-based techniques.} Let us fix a hub labeling for $G=(V,E,w)$ of maximum label size $\leq k$. For every $1 \leq i,j\leq k$, let $V_{i,j} = \{ v \in V \mid |L^-(v)| = i, \ |L^+(v)| = j \}$. In particular, $(V_{i,j})_{1 \leq i,j \leq k}$ is a partition of $V$. For a fixed choice of $i,j$ and any vertex $u$, we want to compute: $e^D(u,V_{i,j}) = \max_{v \in V_{i,j}} D(u,v).$
Indeed, $e^D(u) = \max_{i,j}e^D(u,V_{i,j})$. For that, {\em let us fix a total ordering over $V$}.
For any vertex $u \in V$ and $X \subseteq L^+(u), \ Y \subseteq L^-(u)$, we now define: \begin{align*}e^D(u,V_{i,j},X,Y) = \max\{ D(u,v) \mid v \in V_{i,j}, &\ L^+(u) \cap L^-(v) = X, \\ &\ L^-(u) \cap L^+(v) = Y \}.\end{align*}
We observe that $e^D(u,V_{i,j}) = \max_{X \subseteq L^+(u), \ Y \subseteq L^-(u)} e^D(u,V_{i,j},X,Y)$, that follows from the definition of hub labelings. 
Finally, let $(x,y) \in X \times Y$ be arbitrary. 
Recall that we totally ordered the vertices of the graph. We further reduce the computation of all eccentricities to computing:
\begin{align*}
e^D(u,&V_{i,j},X,Y,x,y) = \max\{ D(u,v) \mid v \in V_{i,j}, \\
& L^+(u) \cap L^-(v) = X, \ L^-(u) \cap L^+(v) = Y, \\
&x \ \text{is the least vertex of} \ X \ \text{s.t.} \ dist(u,v) = dist(u,x) + dist(x,v), \\
&y \ \text{is the least vertex of} \ Y \ \text{s.t.} \ dist(v,u) = dist(v,y) + dist(y,u) \},
\end{align*}
for all $u \in V$ and $\varphi = (V_{i,j},X,Y,x,y)$.
As before we can observe that $e^D(u,V_{i,j},X,Y) = \max_{x \in X,y \in Y}  e^D(u,V_{i,j},X,Y,x,y)$.  
Furthermore, because the maximum label size is $\leq k$, for any fixed vertex $u$ there are at most ${O}(k^44^k) = 2^{{O}(k)}$ tuples $(V_{i,j},X,Y,x,y)$ to consider.
We will show how to reduce the computation of $e^D(u,V_{i,j},X,Y,x,y)$ to $2^{{O}(k)}$ range queries. 
The following technical lemma is the gist of our approach in the paper.

\begin{lemma}\label{lem:red-box}
Let a hub labeling with maximum label size $\leq k$ be given for a graph $G=(V,E,w)$.
In $O(2^{O(k)}|V|)$ time, we can map every vertex $v \in V$ to a $(2(|L^-(v)|+|L^+(v)|)-1)$-dimensional point-set $\{ \overrightarrow{p(v,s,t)} \mid 1 \leq s \leq |L^-(v)|, \ 1 \leq t \leq |L^+(v)| \}$.

\smallskip
Moreover, for every $u \in V$ and $\varphi=(V_{i,j},X,Y,x,y)$, let $V_{u,\varphi}^{\leq}$ (resp., $V_{u,\varphi}^{>}$) be the set of all vertices $v \in V_{i,j}$ s.t.: $L^+(u) \cap L^-(v) = X$, $L^-(u) \cap L^+(v) = Y$, $x$ is the least vertex of $X$ on a shortest $uv$-path, $y$ is the least vertex of $Y$ on a shortest $vu$-path, and $dist(u,v) \leq dist(v,u)$ (resp., $dist(u,v) > dist(v,u)$).
In $2^{O(k)}$ time, we can compute a family of $2^{O(k)}$ boxes $\mathcal{R}\langle \overrightarrow{\ell_X}, \overrightarrow{\ell_Y} \rangle$ where $\overrightarrow{\ell_X} = (\ell_{x^*})_{x^* \in X} \subseteq [i+1]^X$, $\overrightarrow{\ell_Y} = (\ell_{y^*})_{y^* \in Y} \subseteq [j+1]^Y$, and we have:
$$ v \in V_{u,\varphi}^{\leq} \ \text{(resp.,} \ V_{u,\varphi}^{>} \text{)} \ \Longleftrightarrow \exists (\overrightarrow{\ell_X},\overrightarrow{\ell_Y}) \ \text{s.t.} \ \overrightarrow{p(v,\ell_x,\ell_y)} \in  \mathcal{R}\langle \overrightarrow{\ell_X}, \overrightarrow{\ell_Y} \rangle.$$
If $v \in V_{u,\varphi}^{\leq} \ \text{(resp.,} \ V_{u,\varphi}^{>} \text{)}$, then there exists a {\em unique} $(\overrightarrow{\ell_X},\overrightarrow{\ell_Y}) \ \text{s.t.} \ \overrightarrow{p(v,\ell_x,\ell_y)} \in  \mathcal{R}\langle \overrightarrow{\ell_X}, \overrightarrow{\ell_Y} \rangle.$
\end{lemma}

\begin{proof}
For every $v \in V_{i,j}$ and $1 \leq s \leq i, \ 1 \leq t \leq j$, the point $\overrightarrow{p(v,s,t)} = (p_\ell(v,s,t))_{1 \leq \ell \leq 2(i+j)-1}$ is defined as follows: 
\begin{itemize}
\item The elements $p_1(v,s,t),p_2(v,s,t),\ldots,p_i(v,s,t)$ are the vertices in $L^-(v)$ totally ordered. In particular, these $i$ first coordinates form an $i$-dimensional point $\overrightarrow{q^-(v)} = (q^-_\ell(v))_{1 \leq \ell \leq i}$ that is a common prefix to all the points $\overrightarrow{p(v,s',t')}$. 
\item The elements $p_{i+1}(v,s,t),p_{i+2}(v,s,t),\ldots,p_{i+j}(v,s,t)$ are the vertices in $L^+(v)$ totally ordered. These $j$ consecutive coordinates form a $j$-dimensional point $\overrightarrow{q^+(v)} = (q^+_\ell(v))_{1 \leq \ell \leq j}$ that is also common to all the points $\overrightarrow{p(v,s',t')}$. 
\item The elements $p_{i+j+1}(v,s,t),p_{i+j+2}(v,s,t),\ldots,p_{2i+j-1}(v,s,t)$ are equal to the values: $dist_G(q^-_r(v),v) - dist_G(q^-_s(v),v),$ for $1 \leq r \leq i, r \neq s$. In particular, these coordinates may be different between $\overrightarrow{p(v,s,t)}$ and $\overrightarrow{p(v,s',t')}$, for $s \neq s'$.   
\item The elements $p_{2i+j}(v,s,t),p_{2i+j+1}(v,s,t),\ldots, p_{2(i+j-1)}(v,s,t)$ are equal to: $dist_G(v,q^+_r(v)) - dist_G(v,q^+_t(v)),$ for $1 \leq r \leq j, r \neq t$. In particular, these coordinates may be different between $\overrightarrow{p(v,s,t)}$ and $\overrightarrow{p(v,s',t')}$, for $t \neq t'$.    
\item Finally, we set the last coordinate to $p_{2(i+j)-1}(v,s,t) =  dist_G(q^-_s(v),v) - dist_G(v,q^+_t(v))$.
\end{itemize}
Then, let $u \in V$ and $\varphi = (V_{i,j},X,Y,x,y)$ be fixed. In order to restrict ourselves to the vertices $v \in V_{i,j}$ s.t. $L^+(u) \cap L^-(v) = X$, we define a family of $2^{{O}(k)}$ range queries over the points $\overrightarrow{q^-(v)}, \ v \in V_{i,j}$.
\begin{enumerate}
\item We encode the set of $|X|$ indices in which the nodes of $X$ must be found. That is, let $X = (x_1,x_2,\ldots,x_{|X|})$ be totally ordered. For a fixed choice of $1 \leq \ell_1 < \ell_2 < \ldots < \ell_{|X|} \leq i$, we restrict ourselves the vertices $v$ such that $\forall 1 \leq r \leq |X|, \ q^-_{\ell_r}(v) = x_r$. We stress that there are only $\binom{i}{|X|} = 2^{{O}(k)}$ possibilities. 
\item In order to exclude $L^+(u)\setminus X$ from $L^-(v)$, we encode, for each vertex $z$ in this subset, the least index $\ell$ s.t. $q^-_{\ell}(v)$, and so, all subsequent vertices of $L^-(v)$, is greater than $z$; if no such index $\ell$ exists, then by convention we associate to $z$ the value $i+1$. Specifically, for every $z \in L^+(u)\setminus X$, we pick $1 \leq \ell_z \leq i+1$ and we add the following range constraints for our query:
$$\begin{cases}
\forall 1 \leq \ell < \ell_z, \ q^-_{\ell}(v) < z \\
\forall \ell_z \leq \ell \leq i, \ q^-_{\ell}(v) > z
\end{cases}$$
Since we totally ordered $L^+(u)$, each possibility is fully characterized by: \texttt{(i)} the set of indices $\{ \ell_z \mid z \in L^+(u) \setminus X \}$; \texttt{(ii)} and an ordered partition of $L^+(u) \setminus X$ such that two vertices $z,z'$ are in the same group if and only if $\ell_z = \ell_{z'}$. Hence, the number of possibilities here is at most $\sum_{\ell = 1}^{|L^+(u)|-|X|} \binom{i+1}{\ell} \cdot 2^{|L^+(u)|-|X|} \leq 2^{k+1} \cdot 2^{|L^+(u)|-|X|} = 2^{{O}(k)}$.
\end{enumerate} 
Furthermore, in the exact same way as above, in order to restrict ourselves to the points $v \in V_{i,j}$ s.t. $L^-(u) \cap L^+(v) = Y$, we can define a family of $2^{{O}(k)}$ range queries over the points $\overrightarrow{q^+(v)}, \ v \in V_{i,j}$. For the remaining of the proof, let us fix one query over the points $\overrightarrow{q^-(v)}, \ v \in V_{i,j}$, and one query over the points $\overrightarrow{q^+(v)}, \ v \in V_{i,j}$. Note that we can represent the latter as two sequences of indices $\overrightarrow{\ell_X} = (\ell_{x^*})_{x^* \in X}$ and $\overrightarrow{\ell_Y} = (\ell_{y^*})_{y^* \in Y}$. 
We want to further restrict ourselves to the vertices $v$ s.t. $x$ is the least vertex of $X$ on a shortest $uv$-path, and in the same way $y$ is the least vertex of $Y$ on a shortest $vu$-path. The above condition on $x$ is equivalent to have $\forall x' \in X$:
$$\begin{cases}
(1) \ x' < x \Longrightarrow dist_G(u,x') + dist_G(x',v) > dist_G(u,x) + dist_G(x,v) \\
(2) \ x' > x \Longrightarrow dist_G(u,x') + dist_G(x',v) \geq dist_G(u,x) + dist_G(x,v) 
\end{cases}$$
which can be rewritten as:
$$\begin{cases}
(1) \ dist_G(q^-_{\ell_{x'}}(v),v) - dist_G(q^-_{\ell_{x}}(v),v) > \left(dist_G(u,x) - dist_G(u,x')\right) \\
(2) \ dist_G(q^-_{\ell_{x'}}(v),v) - dist_G(q^-_{\ell_{x}}(v),v) \geq \left(dist_G(u,x) - dist_G(u,x')\right) 
\end{cases}$$
and under this form, can be encoded as additional range constraints over the points $\overrightarrow{p(v,\ell_x,\ell_y)}$, for the indices between $i+j+1$ and $2i+j-1$. We proceed similarly for the desired condition on the vertex $y$. Finally, in order to complete these inequalities into a range query for $V_{u,\varphi}^{\leq}$, we further impose: $p_{2(i+j)-1}(v,\ell_x,\ell_y) = dist_G(q^-_{\ell_x}(v),v) - dist_G(v,q^+_{\ell_y}(v)) \leq (dist_G(y,u)-dist_G(u,x)).$
Indeed, in this situation, $dist_G(u,v) = dist_G(u,x) + dist_G(x,v) \leq dist_G(v,y) + dist_G(y,u) = dist_G(v,u)$. We proceed similarly for $V_{u,\varphi}^{>}$ by changing direction of the last inequality. 
\end{proof}

We split the proof for eccentricity queries in two lemmas so as to take into account additional technicalities for the min- and max-distance. Specifically:

\begin{lemma}\label{lem:ecc-sourc-round}
Let a hub labeling with maximum label size $\leq k$ be given for a graph $G=(V,E,w)$.
If $D(\cdot,\cdot)$ is the source distance or the roundtrip distance, then after a pre-processing in $2^{{O}(k)}|V|^{1+o(1)}$ time, for any $u \in V$ and $\varphi=(V_{i,j},X,Y,x,y)$ we can compute $e^D(u,V_{i,j},X,Y,x,y)$ in $2^{{ O}(k)}|V|^{o(1)}$ time.
\end{lemma}

\begin{proof}
We create ${O}(k^4)$ range trees, for point sets of various dimensions.
More specifically, let $i,j \in \{1,\ldots,k\}$ be fixed.
We create $ij = {O}(k^2)$ different $(2(i+j)-1)$-dimensional range trees, that are indexed by all possible pairs $(s,t) \in [i] \times [j]$.
For every vertex $v \in V_{i,j}$, we insert the point $\overrightarrow{p(v,s,t)}$ (defined in Lemma~\ref{lem:red-box}) in the range tree with same index $(s,t)$.
The corresponding value $f_D(\overrightarrow{p(v,s,t)})$ depends on the distance $D(\cdot,\cdot)$ considered. This will be discussed at the end of the proof.
By Lemmas~\ref{lem:range-tree} (applied ${O}(k^4)$ times) and~\ref{lem:red-box}, this overall pre-processing phase can be executed in total ${O}(k^4) \cdot 2^{{O}(k)} \cdot |V|^{1+o(1)} = 2^{{O}(k)} \cdot |V|^{1+o(1)}$ time. 

\medskip
\noindent
\underline{Answering a query.} 
In what follows, let $u \in V$ and $\varphi = (V_{i,j},X,Y,x,y)$ be fixed.
Applying Lemma~\ref{lem:red-box}, we compute $2^{O(k)}$ boxes for the points $\overrightarrow{p(v,s,t)}$ to which the vertices $v \in V_{u,\varphi} = V_{u,\varphi}^{\leq} \cup V_{u,\varphi}^>$ were mapped.
For every such box $\mathcal{R}\langle \overrightarrow{\ell_X},  \overrightarrow{\ell_Y} \rangle$, let $v_{x,y} \in V_{i,j}$ be such that $\overrightarrow{p(v_{x,y},\ell_x,\ell_y)} \in \mathcal{R}\langle \overrightarrow{\ell_X},  \overrightarrow{\ell_Y} \rangle$ and $f_D(\overrightarrow{p(v_{x,y},\ell_x,\ell_y)})$ is maximized.
\begin{itemize}
\item If $D(\cdot,\cdot)$ is the source distance, then we set $\forall (v,s,t) \ f_D(\overrightarrow{p(v,s,t)}) = dist_G(q^-_s(v),v)$. In particular, we have: $e^D(u,V_{i,j},X,Y,x,y) = dist_G(u,x) + \max_{\overrightarrow{\ell_X},  \overrightarrow{\ell_Y}}\max\{ f_D(\overrightarrow{p(v,\ell_x,\ell_y)}) \mid \overrightarrow{p(v,\ell_x,\ell_y)} \in \mathcal{R}\langle \overrightarrow{\ell_X},  \overrightarrow{\ell_Y} \rangle \}$.
Note that in this special case, $Y,y$ are irrelevant.
\item If $D(\cdot,\cdot)$ is the roundtrip distance, then we set $\forall (v,s,t) \ f_D(\overrightarrow{p(v,s,t)}) = dist_G(q^-_s(v),v) + dist_G(v,q^+_t(v)$. In particular, we have:
$e^D(u,V_{i,j},X,Y,x,y) = dist_G(u,x) + dist_G(y,u) 
+ \max_{\overrightarrow{\ell_X},  \overrightarrow{\ell_Y}}\max\{ f_D(\overrightarrow{p(v,\ell_x,\ell_y)}) \mid \overrightarrow{p(v,\ell_x,\ell_y)} \in \mathcal{R}\langle \overrightarrow{\ell_X},  \overrightarrow{\ell_Y} \rangle \}.$
\end{itemize}
We are done by applying Lemma~\ref{lem:range-tree} for max-queries.
\end{proof}

\begin{lemma}\label{lem:ecc-min-max}
Let a hub labeling with maximum label size $\leq k$ be given for a graph $G=(V,E,w)$.
If $D(\cdot,\cdot)$ is either the min-distance or the max-distance, then after a pre-processing in $2^{{O}(k)}|V|^{1+o(1)}$ time, for any $u \in V$ and $\varphi = (V_{i,j},X,Y,x,y)$ we can compute $e^D(u,V_{i,j},X,Y,x,y)$ in $2^{{ O}(k)}|V|^{o(1)}$ time.
\end{lemma}

\begin{proof}
We only detail the necessary modifications for the proof of Lemma~\ref{lem:ecc-sourc-round}. Let $i,j \in \{1,\ldots,k\}$ be fixed. We create $2ij = {O}(k^2)$ different $(2(i+j)-1)$-dimensional range trees, that are indexed by $(s,t,0)$ and $(s,t,1)$ for all possible pairs $(s,t) \in [i] \times [j]$. -- In particular, we need twice more range trees than for Lemma~\ref{lem:ecc-sourc-round}. -- For every $v \in V_{i,j}$ and $1 \leq s \leq i, \ 1 \leq t \leq j$, we insert two identical copies of $\overrightarrow{p(v,s,t)}$ in the range trees that are indexed by $(s,t,0)$ and $(s,t,1)$, with different values associated:
\begin{itemize}
\item For the index $(s,t,0)$, $f_0(\overrightarrow{p(v,s,t)}) = dist_G(q^-_s(v),v)$;
\item For the index $(s,t,1)$, $f_1(\overrightarrow{p(v,s,t)}) = dist_G(v,q^+_t(v))$.
\end{itemize}
Let $u \in V$ and $\varphi=(V_{i,j},X,Y,x,y)$ be fixed. 
Applying Lemma~\ref{lem:red-box}, we compute a family of $2^{O(k)}$ boxes $\mathcal{R}_0\langle \overrightarrow{\ell_X},  \overrightarrow{\ell_Y} \rangle$ for the points $\overrightarrow{p(v,s,t)}$ to which the vertices $v \in V_{u,\varphi}^{\leq}$ were mapped. In the same way, we compute a family of $2^{O(k)}$ boxes $\mathcal{R}_1\langle \overrightarrow{\ell_X},  \overrightarrow{\ell_Y} \rangle$ for the points $\overrightarrow{p(v,s,t)}$ to which the vertices $v \in V_{u,\varphi}^{>}$ were mapped
In doing so, if $D(\cdot,\cdot)$ is the min-distance then  we have: $e^D(u,V_{i,j},X,Y,x,y) = \max_{\overrightarrow{\ell_X},  \overrightarrow{\ell_Y}} \max\{ dist_G(u,x) 
+ f_0(\overrightarrow{p(v,\ell_x,\ell_y)}) \mid \overrightarrow{p(v,\ell_x,\ell_y)} \in \mathcal{R}_0\langle \overrightarrow{\ell_X},  \overrightarrow{\ell_Y} \rangle \} 
 \cup \{ dist_G(y,u) + f_1(\overrightarrow{p(v,\ell_x,\ell_y)}) \mid \overrightarrow{p(v,\ell_x,\ell_y)} \in \mathcal{R}_1\langle \overrightarrow{\ell_X},  \overrightarrow{\ell_Y} \rangle \}$.
It is straightforward to adapt the above to the max-distance, {\it i.e.}, by reversing the respective roles of $\mathcal{R}_0$ and $\mathcal{R}_1$.
\end{proof}

Lemmas~\ref{lem:ecc-sourc-round} and~\ref{lem:ecc-min-max} complete the proof of Theorem~\ref{thm:hub-label} for eccentricity queries. 
For adapting our approach to distance-sum queries, the key observation is that, for any fixed $u \in V$, the sets $V_{u,\varphi}^{\leq}$ and $V_{u,\varphi}^{>}$, over all possible tuples $\varphi=(V_{i,j},X,Y,x,y)$, form a partition of $V$. In particular, for any distance $D(\cdot,\cdot)$ we have: $s^D(u) = \sum_\varphi \left[ \sum_{v \in V_{u,\varphi}^{\leq}} D(u,v) +  \sum_{v \in V_{u,\varphi}^{>}} D(u,v) \right]$.
Then, for a fixed $\varphi$, applying Lemma~\ref{lem:red-box} we compute a family of $2^{O(k)}$ boxes $\mathcal{R}\langle \overrightarrow{\ell_X},\overrightarrow{\ell_Y}\rangle$ for the points $\overrightarrow{p(v,s,t)}$ to which the vertices $v \in V_{u,\varphi}^{\leq}$ were mapped.
We have that $\{ v \in V \mid \overrightarrow{p(v,\ell_x,\ell_y)} \in \mathcal{R}\langle \overrightarrow{\ell_X},\overrightarrow{\ell_Y}\rangle \}$, for all above boxes $\mathcal{R}\langle \overrightarrow{\ell_X},\overrightarrow{\ell_Y}\rangle$, is a partition of $V_{u,\varphi}^{\leq}$. Furthermore, we defined in Lemmas~\ref{lem:ecc-sourc-round} and~\ref{lem:ecc-min-max} some value $D_u$ and function $f_D(\cdot)$ so that $\forall v \in  V_{u,\varphi}^{\leq}, \ \overrightarrow{p(v,\ell_x,\ell_y)} \in \mathcal{R}\langle \overrightarrow{\ell_X},\overrightarrow{\ell_Y}\rangle \Longrightarrow D(u,v) = D_u + f_D(\overrightarrow{p(v,\ell_x,\ell_y)})$. As a result: $\sum_{v \in V_{u,\varphi}^{\leq}} D(u,v) = D_u \cdot |V_{u,\varphi}^{\leq}| + \sum_{(\overrightarrow{\ell_X},\overrightarrow{\ell_Y})}  \sum \{ f_D(\overrightarrow{p(v,\ell_x,\ell_y)}) \mid \overrightarrow{p(v,\ell_x,\ell_y)} \in \mathcal{R}\langle \overrightarrow{\ell_X},\overrightarrow{\ell_Y}\rangle \}$.
The latter computation reduces to a count-query and a sum-query for each box considered.
We proceed similarly for the vertices $v \in V_{u,\varphi}^{>}$.
\qed

%% file: applications.tex
\subsection{Bounded-treewidth graphs}\label{sec:tw}
\input{treewidth}

\subsection{Topological indices}\label{sec:ti}
\input{topological-indices}

%% file: treewidth.tex
A {\em tree-decomposition} for an undirected graph $G$ is a pair $(T,\mathcal{X})$, where $T$ is a tree and $\mathcal{X} = (X_t)_{t \in V(T)}$ is a collection of subsets of $V(G)$ satisfying the following two properties:
\begin{itemize}
\item for every vertex $v$ of $G$, $\{ t \in V(T) \mid v \in X_t\}$ induces a nonempty subtree of $T$;
\item for every edge $uv$ of $G$, there exists a $t \in V(T)$ s.t. $u,v \in X_t$.
\end{itemize}
The {\em width} of a tree-decomposition is equal to $\max_{t \in V(T)} |X_t| - 1$. The {\em treewidth} of an undirected graph $G$ is the minimum width over its tree-decompositions. Finally, the treewidth of a directed graph $G$ is the treewidth of its underlying graph (obtained by removing the arcs orientation)\footnote{There also exist directed variants of treewidth~\cite{JRST01}. We do {\em not} address these variants in the paper.}. 
We refer to~\cite{Bod97} for a compendium of many algorithmic applications of bounded-treewidth graphs, and to~\cite{FLSP+18} for more recent such applications in the field of Fine-Grained complexity in P. Some real-world graphs have bounded-treewidth, {\it e.g.}, the control-flow graphs of well-structured programs in C and various other programming languages~\cite{Tho98}. 
Eccentricity and distance-sum computations for bounded-treewidth graphs have been considered in~\cite{AGW14,BHM20,CaK09}.
In particular, with the notable exception of roundtrip distance, for all other distance functions considered in this article, on any directed graph of order $n$ and treewidth at most $k$ we can compute the eccentricity and the distance-sum of all the vertices in total $2^{{O}(k)}n^{1+o(1)}$ time.
However, if one is interested in the roundtrip distance, then the algorithms proposed in~\cite{AGW14} -- which, to the best of our knowledge, were the best ones known for these two problems before our work -- run in $2^{{ O}(k^2)}n^{1+o(1)}$ time. We replace this quadratic dependency in the treewidth by a linear one, that is optimal under SETH.

\begin{proof}[Proof of Corollary~\ref{cor:treewidth}]
We start by reminding the reader of the overall strategy in~\cite{AGW14}. Let $G$ be a directed graph of order $n$ and treewidth at most $k$.
\begin{enumerate}
\item We compute a tree-decomposition of width ${O}(k)$. It can be done in $2^{{O}(k)}n$ time~\cite{BDDF+16}.
\item We find a node $t \in V(T)$ such that $C = X_t$ is a balanced separator. {\it I.e.}, we can partition $V$ into $A,B,C$ s.t. $\max\{|A|,|B|\} \leq 2n/3$ and all the paths between $A$ and $B$ intersect $C$. It can be computed in linear time if a tree-decomposition is given.
\item Then, for every $a \in A$, we compute $e_A(a) = \max\{dist_G(a,b) + dist_G(b,a) \mid b \in B\}$ and $s_A(a) = \sum \{dist_G(a,b) + dist_G(b,a) \mid b \in B\}$. We proceed similarly for every vertex $b \in B$ ({\it i.e.}, by reversing the roles of $A$ and $B$).
\item Add all possible arcs between the vertices in $C$, so that for every $c,c' \in C$ the arc $cc'$ has weight equal to $dist_G(c,c')$. In doing so, we get a supergraph $G'$. We recurse on $G'[A \cup C]$ and $G'[B \cup C]$ separately.
\end{enumerate}
If Step $3$ can be executed in $T(n,k)$ time, then the total running time of the algorithm is in ${O}(\max\{2^{{O}(k)}n, T(n,k)\}\log{n})$ time~\cite{AGW14}. We prove in what follows that we can solve Step $3$ in $2^{{ O}(k)}n^{1+o(1)}$ time. For that, for every vertex $v$, let us define $L^-(v) = (dist_G(c,v))_{c \in C}$ and $L^+(v) = (dist_G(v,c))_{c \in C}$. These labels can be computed in total $kn^{1+o(1)}$ time. We may not get a hub labeling, however we have the following weaker property: $\forall a \in A, \ b \in B, \ dist_G(a,b) = \min\{dist_G(a,c) + dist_G(c,b) \mid c \in C\}$, and in the same way $dist_G(b,a) = \min\{dist_G(b,c) + dist_G(c,a) \mid c \in C\}$. We are left solving a bi-chromatic version of the eccentricity and distance-sum queries studied in Sec.~\ref{sec:hub}. Up to doubling the number of range trees we use -- in order to store the points corresponding to vertices in $A$ and $B$ separately --, and modifying the range queries of our framework accordingly, our techniques for Theorem~\ref{thm:hub-label} can still be applied to the bi-chromatic versions of these queries.
\end{proof}

%% file: topological-indices.tex
The Wiener index has attracted attention in recent ``fine-grained'' analysis of polynomial-time solvable problems. Indeed, the best-known algorithms for diameter computation on bounded-treewidth graphs and planar graphs can also be applied to the computation of the Wiener index on these graph classes~\cite{BHM20,Cab18}. In this section, we extend this line of work to a broader family of topological indices -- with practical applications to chemistry. Our list excludes important such indices, {\it e.g.}, the so-called Randi\`c index~\cite{Ran75}. This is because the latter is linear-time computable on any graph. By contrast, the naive algorithms for computing all the indices below run in cubic time. We chose to study these specific indices in part because they illustrate some interesting applications of the orthogonal range query framework to distance problems on graphs.

In what follows, {\it we only consider the source distance, for undirected unweighted graphs}. This is because most chemical graphs are undirected. 
Given a graph $G=(V,E)$ and a hub labeling of maximum label size $\leq k$, we may assume $L^-(v) = L^+(v) = L(v)$ for all vertices $v$.
For every $1 \leq i \leq k$, let $V_i = \{ v \in V \mid |L(v)| = i \}$.
The following result simplifies Lemma~\ref{lem:red-box} for undirected graphs:

\begin{lemma}\label{lem:red-box-simpl}
Let a hub labeling with maximum label size $\leq k$ be given for a graph $G=(V,E)$.
In $O(2^{O(k)}|V|)$ time, we can map every vertex $v \in V$ to an $(2|L(v)|-1)$-dimensional point-set $\{ \overrightarrow{p(v,s)} \mid 1 \leq s \leq |L(v)| \}$.
Moreover, for every $u \in V$ and $\varphi=(V_i,X,x)$, let $V_{u,\varphi}$ be the set of all vertices $v \in V_{i}$ s.t.: $L(u) \cap L(v) = X$, and $x$ is the least vertex of $X$ on a shortest $uv$-path.
In $2^{O(k)}$ time, we can compute a family of $2^{O(k)}$ boxes $\mathcal{R}\langle \overrightarrow{\ell_X} \rangle$ where $\overrightarrow{\ell_X} = (\ell_{x^*})_{x^* \in X} \subseteq [i+1]^X$ and:
$$ v \in V_{u,\varphi}  \ \Longleftrightarrow \exists \ \overrightarrow{\ell_X} \ \text{s.t.} \ \overrightarrow{p(v,\ell_x)} \in  \mathcal{R}\langle \overrightarrow{\ell_X} \rangle.$$
If $v \in V_{u,\varphi}$, then there exists a {\em unique} $\overrightarrow{\ell_X} \ \text{s.t.} \ \overrightarrow{p(v,\ell_x)} \in  \mathcal{R}\langle \overrightarrow{\ell_X}\rangle.$
\end{lemma}

Note in particular that for every $1 \leq i \leq k$ and $v \in V_i$, all points $\overrightarrow{p(v,s)}$ (as defined above) start with a common $i$-dimensional prefix $\overrightarrow{q(v)}$, {\it a.k.a.}, the ordered list of all the vertices in $L(v)$.

\paragraph{Wiener index and its relatives.}
Let us recall that the Wiener index of a graph $G$ is equal to $\sum_{u \neq v} dist_G(u,v)$.
The following result is a direct consequence of Theorem~\ref{thm:hub-label} (and in particular, of the distance-sum variant of Lemma~\ref{lem:ecc-sourc-round}):

\begin{corollary}\label{cor:wiener}
For every graph $G = (V,E)$, if we are given a hub labeling with maximum label size $\leq k$, then we can compute its Wiener index in $2^{{O}(k)} \cdot |V|^{1+o(1)}$ time.
\end{corollary}

Several other indices can be computed using similar methods as for the Wiener index.
For instance, the {\em Hyper-Wiener} index of $G$ is equal to $\frac 1 2 \sum_{u \neq v} [ dist_G(u,v) + dist_G(u,v)^2]$~\cite{Ran93}.
By Corollary~\ref{cor:wiener}, we can reduce the computation of the latter to $\sum_{u \neq v} dist_G(u,v)^2$.

\begin{lemma}\label{lem:pow-sum}
For every graph $G = (V,E)$ and integer $\alpha \geq 1$, if we are given a hub labeling with maximum label size $\leq k$, then we can compute $\sum_{u \neq v} dist_G(u,v)^\alpha$ in ${O}(\alpha 2^{{O}(k)} \cdot |V|^{1+o(1)})$ time.
\end{lemma}

\begin{proof}
We store the points $\overrightarrow{p(v,s)}$, for every $1 \leq i \leq k$, $v \in V_i$ and $1 \leq s \leq i$ (as defined in Lemma~\ref{lem:red-box-simpl}), in a family of ${O}(k^2)$ range trees for point-sets of various dimensions in ${O}(k)$.
Each such tree is indexed by a pair $(i,s)$.
Then, let $u \in V$ be fixed.
We have $\sum_v dist(u,v)^{\alpha} = \sum_{\varphi=(V_i,X,x)}\sum_{v \in V_{u,\varphi}} dist(u,v)^{\alpha}$.
Therefore, in what follows, let $\varphi=(V_i,X,x)$ be fixed.
Applying Lemma~\ref{lem:red-box-simpl}, we compute a family of $2^{O(k)}$ rectangles $\mathcal{R}\langle \overrightarrow{\ell_X}\rangle$ s.t. $\{ v \in V \mid \overrightarrow{p(v,\ell_x)} \in \mathcal{R}\langle \overrightarrow{\ell_X}\rangle \}$ partitions $V_{u,\varphi}$.
Furthermore, for every $v \in V_{u,\varphi}$, we have:
$$dist_G(u,v)^\alpha = \sum_{t=0}^\alpha \binom \alpha t dist_G(u,x)^t dist_G(x,v)^{\alpha -t}.$$
We define $\alpha+1$ functions $f_t$ over the point-sets, $0 \leq t \leq \alpha$, so that $f_t(\overrightarrow{p(v,s)}) = dist_G(q_s(v),v)^t$. Doing so, we obtain: $\sum_{v \in V_{u,\varphi}} dist(u,v)^{\alpha} = \sum_{\overrightarrow{\ell_X}}\sum_{t=0}^\alpha \Big[\binom \alpha t dist_G(u,x)^t\Big] \cdot \sum \big\{ f_{\alpha-t}(\overrightarrow{p(v,\ell_x)}) \mid \overrightarrow{p(v,\ell_x)} \in \mathcal{R}\langle\overrightarrow{\ell_X}\rangle \big\}.$
As a result, $\sum_{v \in V_{u,\varphi}} dist(u,v)^{\alpha}$ can be derived from $2^{{O}(k)}$ sum-queries, applied $\alpha+1$ times to different families of range trees. We stress that all these families of range trees contain the same point-sets, but with different values, chosen according to $f_t$.  
\end{proof}

\begin{corollary}\label{cor:hyp-w}
For every graph $G = (V,E)$, if we are given a hub labeling with maximum label size $\leq k$, then we can compute its Hyper-Wiener index in $2^{{O}(k)} \cdot |V|^{1+o(1)}$ time.
\end{corollary}

The Schultz molecular topological index of $G$ (MTI) is equal to $\sum_{u \neq v} \left( d(u) + d(v) \right) \cdot dist_G(u,v)$, where $d(u)$ and $d(v)$ are the degrees of vertices $u$ and $v$, respectively~\cite{Sch89}.

\begin{lemma}\label{lem:mti}
For every graph $G = (V,E)$, if we are given a hub labeling with maximum label size $\leq k$, then we can compute its MTI in ${O}(|E| + 2^{{O}(k)} \cdot |V|^{1+o(1)})$ time.
\end{lemma}

\begin{proof}
Let $u \in V$ and $\varphi = (V_i,X,x)$ be fixed.
We have:
\begin{align*}
\sum_{v \in V_{u,\varphi}} (d(u)+d(v)) &\cdot dist(u,v) \\
&= \sum_{v \in V_{u,\varphi}} (d(u)+d(v)) \cdot [dist(u,x) + dist(x,v)] \\
&= d(u)dist(u,x) \cdot |V_{u,\varphi}| + d(u) \cdot \sum_{v \in V_{u,\varphi}} dist(x,v) \\
&+ d(u,x) \cdot \sum_{v \in V_{u,\varphi}} d(v) + \sum_{v \in V_{u,\varphi}} d(v)dist(x,v).
\end{align*}
Let us define four functions over the point-sets so that: $f_0(\overrightarrow{p(v,s)}) = 1$ (counting), $f_1(\overrightarrow{p(v,s)}) = dist_G(q_s(v),v)$, $f_2(\overrightarrow{p(v,s)}) = d(v)$, $f_3(\overrightarrow{p(v,s)}) = d(v)dist_G(q_s(v),v)$. These functions require the knowledge of the degree sequence of the graph, that can be computed in ${O}(|V|+|E|)$ time. In order to compute the MTI of $G$, we can apply the same multi-valued range query framework as for Lemma~\ref{lem:pow-sum} w.r.t. the functions $f_0,f_1,f_2,f_3$.
\end{proof}

\paragraph{Harary index and beyond.} Another important topological index is the Harary index, defined for any graph $G$ as $\sum_{u \neq v} \frac 1 {dist(u,v)}$~\cite{PNTM93}. The Reciprocal complementary Wiener index of a graph $G$ is equal to $\sum_{u \neq v} \frac 1 {diam(G) + 1 - dist(u,v)}$, where $diam(G)$ is the diameter~\cite{IIB02}. In a recent paper~\cite{Cab19}, Cabello designed an efficient batched evaluation of rational functions, that he applied to the computation of the Harary index on bounded-treewidth graphs. Namely, as in Sec.~\ref{sec:range-query}, let $V$ be a set of $k$-dimensional points, where each point $\overrightarrow{p} = (p_1,p_2,\ldots,p_k) \in V$ is assigned some value $f(\overrightarrow{p})$. An {\em inverse shifted-weight query} (ISW) is given by a pair $(\mathcal{R},\delta)$, for some box $\mathcal{R}$ and non-negative value $\delta$. It outputs $\sum\limits_{\overrightarrow{p} \in V \cap \mathcal{R}} \frac 1 {\delta + f(\overrightarrow{p})}$.

\begin{lemma}[\cite{Cab19}]\label{lem:isw}
Let $V$ be a set of $n$ points in $k$ dimensions, and assume that each point $\overrightarrow{p}$ has a positive value $f(\overrightarrow{p})$.
Consider ISW queries specified by $(\mathcal{R}_1,\delta_1), (\mathcal{R}_2,\delta_2), \ldots, (\mathcal{R}_m,\delta_m)$. These values can be computed in total ${O}(2^{{O}(k)}(n+m)^{1+o(1)})$ time.
\end{lemma}

\begin{lemma}\label{lem:inv}
For every graph $G = (V,E)$, if we are given a hub labeling with maximum label size $\leq k$, then we can compute both its Harary index and its Reciprocal complementary Wiener index in $2^{{O}(k)} \cdot |V|^{1+o(1)}$ time.
\end{lemma}

\begin{proof}
Let us first discuss the case of the Harary index.
For every $v \in V$ and $1 \leq s \leq |L(v)|$, we augment the point $\overrightarrow{p(v,s)}$ (as defined in Lemma~\ref{lem:red-box-simpl}) with one more coordinate, set equal to $v$.
In particular, we map $v$ to an $2|L(v)|$-dimensional point-set.
We create $O(k^2)$ range trees of various dimensions in $O(k)$, that are indexed by $(i,s)$ for every $1 \leq i \leq k$ and $1 \leq s \leq i$.
For every such pair $(i,s)$, the corresponding range tree is $2i$-dimensional; we insert in it all points $\overrightarrow{p(v,s)}$ s.t. $v \in V_i$ {\em and} $q_s(v) \neq v$ (the $s^{th}$ vertex of $L(v)$ is different from $v$ itself). This is because we want to assign to this point the value $f( \overrightarrow{p(v,s)} ) = dist(q_s(v),v)$, that needs to be positive so that we can apply Lemma~\ref{lem:isw} later in the proof.

Then, let $u \in V$ be fixed.
We consider all possible tuples $\varphi = (V_i,X,x)$. In particular:
\begin{align*}
\sum_{u \neq v} \frac 1 {dist(u,v)} &= \sum_{\varphi} \sum_{v \in V_{u,\varphi} \setminus \{u\}} \frac 1 {dist(u,v)} = \sum_{x \in L(u) \setminus \{u\}} \frac 1 {dist(u,x)} \\
&+ \sum_{\varphi} \sum_{v \in V_{u,\varphi} \setminus (L(u) \cup \{u\})} \frac 1 {dist(u,v)}.
\end{align*}
Obviously, $\sum_{x \in L(u) \setminus \{u\}} \frac 1 {dist(u,x)}$ can be computed in ${O}(k)$ time from the hub label. 
Furthermore, for any fixed $\varphi = (V_i,X,x)$, applying Lemma~\ref{lem:red-box-simpl}, we compute a family of $2^{O(k)}$ boxes $\mathcal{R}\langle \overrightarrow{\ell_X}\rangle$ s.t. $\{v \in V \mid \overrightarrow{p(v,\ell_x)} \in  \mathcal{R}\langle \overrightarrow{\ell_X}\rangle \}$ partitions $V_{u,\varphi}$.
However, in order to avoid over-counting, we need to exclude all vertices of $L(u) \cup \{u\}$ from the latter set.
This was already done implicitly for vertex $x$, since having $dist(x,x) = 0$ prevented the corresponding points to be inserted in the range trees.
Nevertheless, we may also have $dist(u,x') = dist(u,x) + dist(x,x')$ for some $x' \in L(u) \setminus \{x\}$. 
Let $\tilde{L}(u) = L(u) \cup \{u\}$.
We replace each box $\mathcal{R}\langle \overrightarrow{\ell_X}\rangle$ by $|\tilde{L}(u)|+1$ boxes, that are obtained by imposing further constraints on the last coordinate.
Specifically, let $q_0(u) = -\infty$, $q_{|\tilde{L}(u)|+1}(u) = +\infty$, and for every $1 \leq s \leq |\tilde{L}(u)|$, let $q_s(u)$ be the $s^{th}$ vertex of $\tilde{L}(u)$.
For every $0 \leq t \leq |\tilde{L}(u)|$, the box $\mathcal{R}_t\langle \overrightarrow{\ell_X}\rangle$ further imposes $p_{2i}(v,\ell_x) \in (q_t(v),q_{t+1}(v))$.
In doing so, we can now partition $V_{u,\varphi} \setminus (L(u) \cup \{u\})$.
Let $\mathcal{R}_1, \mathcal{R}_2, \ldots, \mathcal{R}_m$ be the resulting family of boxes, for some $m = O(k2^{{O}(k)}) = 2^{O(k)}$. In order to compute $\sum_{v \in V_{u,\varphi} \setminus (L(u) \cup \{u\})} \frac 1 {dist(u,v)} = \sum_{v \in V_{u,\varphi} \setminus (L(u) \cup \{u\})} \frac 1 {dist(u,x) + dist(x,v)}$, it suffices to solve the ISW queries $(\mathcal{R}_1,dist(u,x)), (\mathcal{R}_2,dist(u,x)), \ldots, (\mathcal{R}_m,dist(u,x))$.
Overall, we define $2^{{O}(k)}$ ISW queries per vertex. By Lemma~\ref{lem:isw}, we can solve all queries in total $2^{{O}(k)}|V|^{1+o(1)}$ time.

\medskip
Finally, we discuss the changes to bring to the above strategy so that we can also compute the Reciprocal complementary Wiener index. First, we need to compute the diameter, which by Theorem~\ref{thm:hub-label} can be done in $2^{{O}(k)}|V|^{1+o(1)}$ time. The value of a point $\overrightarrow{p(v,s)}$ is now defined as $f( \overrightarrow{p(v,s)} ) = - dist(q_s(v),v)$. Furthermore, given a vertex $u \in V$ and a tuple $\varphi = (V_i,X,x)$, we specify as ISW queries $(\mathcal{R}_1,diam(G)+1-dist(u,x)), (\mathcal{R}_2,diam(G)+1-dist(u,x)), \ldots, (\mathcal{R}_m,diam(G)+1-dist(u,x))$, where $\mathcal{R}_1, \mathcal{R}_2, \ldots, \mathcal{R}_m$ is the same family of $m = 2^{{O}(k)}$ hyper-rectangles as the one described earlier in the proof for computing the Harary index. We {\em cannot} apply Lemma~\ref{lem:isw} directly because we assigned negative values to the points $\overrightarrow{p(v,s)}$. However, a careful analysis of the proof of Lemma~\ref{lem:isw} in~\cite{Cab19} shows that we only need to ensure that we never evaluate a rational function at a singularity (which is always the case under the sufficient conditions stated in the lemma). In the case of the Reciprocal complementary Wiener index, the rational functions to be evaluated are of the form $x \to \frac 1 {x - dist(x,v)}$. They are evaluated at $x = diam(G) +1 - dist(u,x)$ for some vertices $u$ s.t. $dist(u,v) = dist(u,x) + dist(x,v)$. In particular, since we always have $diam(G) +1 - dist(u,v) > 0$, we never need to evaluate these rational functions at a singularity.
\end{proof}

\paragraph{Padmakar-Ivan index and Szeged index.} For the last two indices which we study in this section, we need a few more terminology. For every two vertices $u$ and $v$ in $G$, we denote $n_{uv}$ the number of vertices that are closer to $u$ than to $v$. The Szeged index of $G$ is defined as $\sum_{uv \in E} n_{uv}n_{vu}$~\cite{Gut94}. Similarly, let $n^e_{uv}$ be the number of edges that are closer to $u$ than to $v$ (where the distance between an edge $e=xy$ and a vertex $u$ is defined as $\min\{dist(x,u),dist(y,u)\}$). The Padmakar-Ivan index of $G$ is defined as $\sum_{uv \in E} \left[ n^e_{uv} + n^e_{vu} \right]$~\cite{KKA01}.

\begin{lemma}\label{lem:new}
For every graph $G = (V,E)$, if we are given a hub labeling with maximum label size $\leq k$, then after a pre-processing in $2^{{O}(k)}|V|^{1+o(1)}$ time, for any vertices $u$ and $v$ we can compute the value $n_{uv}$ in $2^{{O}(k)}|V|^{o(1)}$ time.
\end{lemma}

\begin{proof}
Consider the point-sets defined in Lemma~\ref{lem:red-box-simpl}.
For every $v \in V$ and $1 \leq s,t \leq |L(v)|$, we define a point $\overrightarrow{p(v,s,t)}$ which is obtained by concatenating $\overrightarrow{p(v,s)}$ with $\overrightarrow{p(v,t)}$.
Observe that all these points have a dimension in $O(|L(v)|) = O(k)$. 
If $v \in V_i$, then we store the point $\overrightarrow{p(v,s,t)}$ in some range tree indexed by $(i,s,t)$. Overall, there are ${ O}(k^3)$ range trees.

\medskip
\noindent
\underline{Answering a query.} In what follows, let $u,v \in V$ be fixed. For every $i \in \{1,\ldots,k\}$, $X \subseteq L(u), Y \subseteq L(v)$ and $(x,y) \in X \times Y$, we get a tuple $\varphi = (V_i,X,Y,x,y)$ and we define:
\begin{align*}
V_{u,\varphi} = \{ z \in V_i \mid& L(u) \cap L(z) = X, \ L(v) \cap L(z) = Y, \\
& \ x \ \text{is the least vertex of} \ X \ \text{on a shortest} \ uz\text{-path}, \\
& \ y \ \text{is the least vertex of} \ Y \ \text{on a shortest} \ vz\text{-path} \}.
\end{align*}
Let $n_{uv}(\varphi)$ be the number of vertices in $V_{u,\varphi}$ that are closer to $u$ than to $v$. In particular, $n_{uv} = \sum_{\varphi} n_{uv}(\varphi)$. Therefore, we only need to explain how to compute $n_{uv}(\varphi)$ for any fixed tuple $\varphi = (V_i,X,Y,x,y)$. For that, we first adapt the proof of Lemma~\ref{lem:red-box} in order to compute a family of $2^{O(k)}$ boxes $\mathcal{R}\langle \overrightarrow{\ell_X},\overrightarrow{\ell_Y}\rangle$ with the following properties: $z \in V_{u,\varphi} \Longleftrightarrow \exists (\overrightarrow{\ell_X},\overrightarrow{\ell_Y}) \ \text{s.t.} \ \overrightarrow{p(z,\ell_x,\ell_y)} \in \mathcal{R}\langle \overrightarrow{\ell_X},\overrightarrow{\ell_Y}\rangle$; furthermore, if $z \in V_{u,\varphi}$ then there is a {\em unique} $(\overrightarrow{\ell_X},\overrightarrow{\ell_Y}) \ \text{s.t.} \ \overrightarrow{p(z,\ell_x,\ell_y)} \in \mathcal{R}\langle \overrightarrow{\ell_X},\overrightarrow{\ell_Y}\rangle$. 

Let $r_{\varphi}(u,v) = dist(v,y) - dist(u,x)$. For every $z \in V_{u,\varphi}$ we have:
\begin{align*}
&dist(u,z) < dist(v,z) \Longleftrightarrow \ dist(u,x) + dist(x,z) < dist(v,y) + dist(y,z) \\
&\Longleftrightarrow \  dist(x,z) - dist(y,z) < \left[ dist(v,y) - dist(u,x) \right] \\
&\Longleftrightarrow \  dist(q_{\ell_x}(z),z) - dist(q_{\ell_y}(z),z) < r_{\varphi}(u,v),
\end{align*}
which can be verified by modifying correspondingly the range constraints defining $\mathcal{R}\langle \overrightarrow{\ell_X},\overrightarrow{\ell_Y}\rangle$. We are left performing sum-queries over the family of all modified boxes.
\end{proof}

Finally, we show that our above strategy for computing the values $n_{uv}$ can be generalized to the computation of the values $n^e_{uv}$. 

\begin{lemma}\label{lem:new-e}
For every graph $G = (V,E)$, if we are given a hub labeling with maximum label size $\leq k$, then after a pre-processing in $2^{O(k)}|E|^{1+o(1)}$ time, for any vertices $u$ and $v$ we can compute the value $n^e_{uv}$ in $2^{ O(k)}|E|^{o(1)}$ time.
\end{lemma}

\begin{proof}
We modify the framework of Lemma~\ref{lem:new} so that it applies to edges rather than to vertices.
For that, recall that we totally ordered the vertex-set $V$.
In particular, every edge $e$ has a canonical representation $uv$, where $u < v$.
For every $1 \leq i,j \leq k$, let $E_{i,j} = \{ e =uv \in E \mid u \in V_i, \ v \in V_j \}$.
If $e = uv \in E_{i,j}$ then, let $(\overrightarrow{p(u,s,t)})_{1 \leq s,t \leq i}$ and $(\overrightarrow{p(v,s',t')})_{1 \leq s',t' \leq j}$ be the points corresponding to $u$ and $v$ in the proof of Lemma~\ref{lem:new}; we define a family of $(ij)^2 = {O}(k^4)$ points $\overrightarrow{p(e,s,t,s',t')}$ that are obtained from the concatenation of $\overrightarrow{p(u,s,t)}$ with $\overrightarrow{p(v,s',t')}$ for every $1 \leq s,t \leq i$ and $1 \leq s',t' \leq j$. In doing so, we can adapt the proof of Lemma~\ref{lem:new} in order to compute, for every $u,v \in V$ and $\varphi = (V_i,V_j,X,X',Y,Y',x,x',y,y')$ the value:
\begin{align*}
n^e_{uv}(\varphi) = \sum \{ e = zz' &\in E_{i,j} \mid \ L(u) \cap L(z) = X, \ L(v) \cap L(z) = Y, \\
 & \ L(u) \cap L(z') = X', \ L(v) \cap L(z') = Y', \\
 & \ x \ \text{is the least vertex of} \ X \ \text{on a shortest} \ uz\text{-path} \\
 & \ y \ \text{is the least vertex of} \ Y \ \text{on a shortest} \ vz\text{-path} \\
 & \ x' \ \text{is the least vertex of} \ X' \ \text{on a shortest} \ uz'\text{-path} \\
 & \ y' \ \text{is the least vertex of} \ Y' \ \text{on a shortest} \ vz'\text{-path} \\
 & \ dist(u,e) < dist(v,e)\}
\end{align*}
by reducing to $2^{{O}(k)}$ sum-queries.
We are done, as for any $u$ and $v$ we have $n^e_{uv} = \sum_{\varphi} n^e_{uv}(\varphi)$, and there are $2^{O(k)}$ tuples $\varphi$ to be considered.
\end{proof}

The following result now follows from Lemmas~\ref{lem:new} and~\ref{lem:new-e}.

\begin{corollary}\label{cor:new}
For every graph $G = (V,E)$, if we are given a hub labeling with maximum label size $\leq k$, then we can compute its Padmakar-Ivan and Szeged indices in $2^{{O}(k)} \cdot (|V|+|E|)^{1+o(1)}$ time.
\end{corollary}

%% file: expansion.tex
A graph $H$ is an {\em $r$-shallow minor} of a graph $G$ if it can be obtained from some subgraph of $G$ by the contraction of pairwise disjoint subgraphs of radius at most $r$; a graph family $\mathcal{G}$ has {\em bounded expansion} if, for any $r$, every $r$-shallow minor for any graph in $\mathcal{G}$ has edge density at most some constant $d_r(\mathcal{G})$~\cite{NeO12}. A remarkable property of graph classes of bounded expansion, which we use in this section, is the existence of low tree-depth decompositions.

\smallskip
\noindent
Specifically, an undirected graph $G=(V,E)$ has tree-depth at most $k$ if there exists a rooted tree $T$ of height $\leq k$ such that $V(T) = V$, and if $u,v \in V$ are adjacent in $G$ then they have an ancestor-descendent relationship in $T$~\cite{NeO12}. A directed graph has tree-depth at most $k$ if its underlying graph (obtained by removing the orientation of its arcs) has tree-depth at most $k$.

\begin{lemma}[\cite{RRVS14}]\label{lem:fpt-td}
For every fixed $k$ and $n$-vertex graph $G$, we can decide whether the tree-depth of $G$ is at most $k$, and if so compute a corresponding rooted tree $T$, in $2^{{O}(k^2)}n$ time.
\end{lemma} 

\begin{lemma}[\cite{IOO18}]\label{lem:hub-td}
For every fixed $k$ and $n$-vertex graph $G$, if $G$ has tree-depth at most $k$ and a corresponding rooted tree $T$ is given, then we can compute hub labels of size at most $k$ in ${O}(k^2n)$ time.
\end{lemma}

A {\em low tree-depth decomposition} with parameter $k$ for a graph $G$ is a vertex-coloring such that the union of any $i \leq k$ color classes induces a subgraph of tree-depth at most $i$~\cite{NeO12}. We denote by $\chi_k(G)$ the minimum number of colours in a low tree-depth decomposition with parameter $k$ for $G$ (for $k=1$, this is exactly the chromatic number of $G$).

\begin{lemma}[\cite{NeO12}]\label{lem:td-dec}
A graph class $\mathcal{G}$ has bounded expansion if and only if $\forall k, \ \limsup \{ \chi_k(G) \mid G \in \mathcal{G} \} < \infty$.
\end{lemma}

Before presenting our proof of Theorem~\ref{thm:nowhere-dense}, what now remains to discuss is the computational cost for computing a low tree-depth decomposition for graph classes of bounded expansion. For that, for any graph $G$, let $\nabla_r(G)$ be the maximum density of an $r$-shallow minor of $G$. The following result is a straightforward consequence of the definition of graph classes of bounded expansion.

\begin{lemma}[\cite{NeO12}]\label{lem:grad}
A graph class $\mathcal{G}$ has bounded expansion if and only if $\forall k, \ \limsup\{ \nabla_k(G) \mid G \in \mathcal{G} \} < \infty$.
\end{lemma}

We will use the following constructive result of Ne{\v{s}}et{\v{r}}il and Ossona de Mendez for low tree-depth decompositions:

\begin{lemma}[\cite{NeO12}]\label{lem:compute-dec}
For every graph $G$ and integer $k$, there exists a polynomial $P_k$ of degree $\approx 2^{2^k}$ such that we can compute a low tree-depth decomposition with parameter $k$ for $G$ with $P_k(\nabla_{2^{k-2}+1}(G))$ colours, in total ${O}(P_k(\nabla_{2^{k-2}+1}(G)) \cdot |G|)$ time.
\end{lemma}

We are now ready to present the main result in this section:

\begin{proof}[Proof of Theorem~\ref{thm:nowhere-dense}]
Let $c = P_k(\nabla_{2^{k-2}+1}(G))$.
By Lemma~\ref{lem:compute-dec}, a low tree-depth decomposition with parameter $k$ for $G$ and $c$ colours can be computed in total ${O}(c \cdot n)$ time. Consider all possible unions of $k$ color classes sequentially (there are $\binom{c}{k}$ possibilities). Every one of them induces a subgraph of tree-depth at most $k$. In particular, by Lemma~\ref{lem:fpt-td}, we can compute a corresponding rooted tree for this subgraph in $2^{{O}(k^2)}$ time. Then, by Lemma~\ref{lem:hub-td}, we can compute a hub labeling for this subgraph, with maximum label size $\leq k$, in ${O}(k^2n)$ time. By concatenating all these hub labelings, each vertex of $G$ gets assigned a label of cardinality $\leq k \binom c k$. This resulting labeling is not necessarily a hub labeling, as we may not have $dist_G(u,v) = \min \{ dist_G(u,x) + dist_G(x,v) \mid x \in L^+(u) \cap L^-(v) \}$ for all ordered pairs of vertices $(u,v)$. However, since every path of length at most $k$ is contained in the union of at most $k$ color classes (and so, in one of the $\binom{c}{k}$ subgraphs of bounded tree-depth on which our above algorithm iterates), if $dist_G(u,v) \leq k$, then given $u$ and $v$ this above labeling correctly computes $dist_G(u,v)$. In particular, if the diameter is at most $k$ then we indeed computed a hub labeling for $G$. Finally, we use this labeling as input for the algorithm of Theorem~\ref{thm:hub-label} in order to decide in $2^{{O}\left(k\binom{c}{k}\right)}n^{1+o(1)}$ time whether the diameter of $G$ is at most $k$.

The total running time is in $2^{{O}\left(k\binom{c}{k}\right)}n^{1+o(1)}$. Note that by Lemma~\ref{lem:grad}, $c \approx \left[d_{2^{k-2}+1}(\mathcal{G})\right]^{2^{2^k}}$ for some constant $d_{2^{k-2}+1}(\mathcal{G})$ depending only on $\mathcal{G}$ and $k$.
\end{proof}